\documentclass[a4paper,UKenglish,cleveref,autoref,thm-restate]{oasics-v2021}

\usepackage[utf8]{inputenc}
\usepackage{xcolor}

\definecolor{keywordcolor}{rgb}{0.7, 0.1, 0.1}   
\definecolor{tacticcolor}{rgb}{0.0, 0.1, 0.6}    
\definecolor{commentcolor}{rgb}{0.4, 0.4, 0.4}   
\definecolor{symbolcolor}{rgb}{0.0, 0.1, 0.6}    
\definecolor{sortcolor}{rgb}{0.1, 0.5, 0.1}      
\definecolor{attributecolor}{rgb}{0.7, 0.1, 0.1} 
\definecolor{white}{rgb}{1.0,1.0,1.0}

\usepackage{xspace}
\usepackage{nicefrac}
\usepackage{caption}
\usepackage{tikz}

\usepackage[inline,shortlabels]{enumitem} 
\newlist{inlinelist}{enumerate*}{1}
\setlist*[inlinelist,1]{%
  label=(\roman*),
}

\hypersetup{
    colorlinks=true,  
    urlcolor=blue,
}

\usepackage[final,nomargin,inline,index]{fixme} 
\fxusetheme{color}
\FXRegisterAuthor{bart}{anbart}{\color{magenta} {\underline{bart}}}
\FXRegisterAuthor{dani}{andani}{\color{red} {\underline{dani}}}

\usepackage{cleveref}
\crefname{table}{Listing}{Listings}  

\newcommand{\codefont}{\fontsize{10}{10}\selectfont}
\newcommand{\code}[1]{{\tt\codefont {#1}}}

\newcommand{\eg}{e.g.\@\xspace}
\newcommand{\ie}{i.e.\@\xspace}

\title{Formalizing Automated Market Makers in the Lean~4 Theorem Prover}

\titlerunning{Formalizing Automated Market Makers in the Lean 4 Theorem Prover} 
\author{Daniele Pusceddu}{ETH Zurich, Switzerland \and University of Cagliari, Italy }{dpusceddu@student.ethz.ch}{}{}

\author{Massimo Bartoletti}{University of Cagliari, Italy \and \url{http://blockchain.unica.it} }{bart@unica.it}{https://orcid.org/0000-0003-3796-9774}{Partially supported by project SERICS (PE00000014) and PRIN 2022 DeLiCE (F53D23009130001) under the MUR National Recovery and Resilience Plan funded by the European Union -- NextGenerationEU.}

\authorrunning{D. Pusceddu and M. Bartoletti} 

\Copyright{Daniele Pusceddu and Massimo Bartoletti} 



\category{} 

\relatedversion{} 

\nolinenumbers 

\begin{document}

\maketitle

\begin{abstract}
Automated Market Makers (AMMs) are an integral component of the decentralized finance (DeFi) ecosystem, as they allow users to exchange crypto-assets 
without the need for trusted authorities or external price oracles. 
Although these protocols are based on relatively simple mechanisms, 
\eg to algorithmically determine the exchange rate between crypto-assets, 
they give rise to complex economic behaviours. 
This complexity is witnessed by the proliferation of models that study their structural and economic properties. 
Currently, most of theoretical results obtained on these models 
are supported by pen-and-paper proofs.
This work proposes a formalization of constant-product AMMs in the Lean~4 Theorem Prover.
To demonstrate the utility of our model, 
we provide mechanized proofs of key economic properties
like arbitrage, that at the best of our knowledge have only been proved 
by pen-and-paper before.
\end{abstract}

\section{Introduction}
\label{sec:intro}

Automated Market Makers (AMMs) are one of the key applications in the Decentralized Finance (DeFi) ecosystem,
as they allow users to trade crypto-assets without the need for trusted intermediaries~\cite{Xu23csur}. 
Unlike traditional order-book exchanges, where buyers and sellers must find a counterpart, AMMs enable traders to autonomously swap assets deposited in liquidity pools contributed by other users, who are incentivized to provide liquidity by a complex reward mechanism. 
At the time of writing, there are multiple AMM protocols controlling several billions of dollars worth of assets%
\footnote{\url{https://defillama.com/protocols/Dexes}}. 
This has made AMMs an appealing target for attacks, resulting in losses worth billions of dollars over time\footnote{\url{https://chainsec.io/defi-hacks/}}.

The security of AMMs depends on several factors: besides the absence of traditional programming bugs, it is crucial that their economic mechanism gives rise to a rational behaviour of its users that aligns with the AMM ideal functionality, \ie providing an algorithmic exchange rate coherent with the one given by trusted price oracles. 
Therefore, it is important to obtain strong guarantees about the economic mechanisms of these protocols. 
While formal verification tools for smart contracts based on model-checking are useful in detecting programming bugs and even in proving some structural properties of AMMs~\cite{certora-compound-v2,certora-compound-v3}, they are not suitable for verifying, or even expressing more complex properties regarding the economic mechanism of AMMs.
These economic mechanisms have been studied in several research works, which, in most cases, provide pen-and-paper proofs of the obtained properties.
Given the complexity of the studied models, it would be desirable to also provide machine-verified proofs, so that we may rely on the proven properties beyond any reasonable doubt. 
To the best of our knowledge, existing mechanized formalizations~\cite{Nielsen23cpp} focus on verifying relevant structural properties of AMMs like their state consistency, and not on studying the economic mechanism of AMMs (see~\Cref{sec:related} for a detailed comparison).

\paragraph*{Contributions}
In this paper we formalize Automated Market Makers in the Lean 4 theorem prover.
Our model is based on a slightly simplified version of the Uniswap v2 protocol (one of the leading AMMs), which was studied in~\cite{Bartoletti_2022} with a pen-and-paper formalization. We provide a Lean specification of blockchain states, abstracted from any factors that are immaterial to the study of AMMs. Then, we model the fundamental interactions that users may have with AMMs as well as the economic notions of price, networth and gain. Finally, we build machine-checked proofs of economic properties of constant-product AMMs. In particular, we derive an explicit formula for the economic gain obtained by a user after an exchange with an AMM. Building upon this formula we prove that, from a trader's perspective, aligning a constant-product AMM's internal exchange rate with the rate given by the trader's price oracle implies the optimal gain from that AMM. This results in the fundamental property of AMMs acting as price oracles themselves~\cite{Angeris20aft}. We then construct the optimal swap transaction that a rational user can perform to maximize their gain, solving the arbitrage problem. Our formalization and proofs\footnote{\url{https://github.com/danielepusceddu/lean4-amm}} are made available in a public GitHub repository.
At the best of our knowledge, this is the first mechanized formalization of the economic mechanism of AMMs.
We finally discuss some open issues, and alternative design choices for formalizing AMMs.

\section{Formalization}
\label{sec:blockchain}

An Automated Market Maker implements a decentralized exchange between two different token types. The exchange rate is determined by a smart contract, which also takes care of performing the exchange itself: namely, the contract receives from a trader some amount of the input token type, and sends back the correct amount of the output token type, which is taken from the AMM reserves. A single smart contract can control many instances of AMMs (also called \emph{AMM pairs}): we may have a pair for each possible unordered pair of token types. To create an AMM instance, a user must provide the initial liquidity for the reserves of that pair of tokens. Liquidity providers are rewarded with a type of token that specifically represents shares in that AMM's reserves: we call these \emph{minted} token types, while any other token type will be called \emph{atomic}.

\paragraph*{Blockchain State}

We begin by formalizing the blockchain state, abstracting from the details that are immaterial to the study of AMMs. 
Then, our model includes the users' wallets, the AMMs and their reserves
(see~\Cref{tab:amm-system}).
We formalize the universes of users and atomic token types as the types \code{A} and \code{T}, respectively, as structures that encapsulate a natural number.
Hereafter, we use $a, b, \ldots$ to denote users in $\code{A}$, 
and $\tau, \tau_0, \tau_1, \ldots$ to denote tokens in $\code{T}$.
Minted token types are pairs of \code T. 
We represent the funds owned by a user by a \emph{wallet} that maps token types to non-negative reals.
To rule out wallets with infinite tokens,
we use Mathlib's finitely supported functions\footnote{\url{https://leanprover-community.github.io/mathlib4_docs/Mathlib/Data/Finsupp/Defs.html}}: 
in general,
given any type $\alpha$ and any type $M$ with a $0$ element, $f \in \alpha \to_0 M$ if $\mathrm{supp}\left(f\right) = \left\{x \in \alpha \mid f\left(x\right) \neq 0\right\}$ is finite.  

We define the type $\code W_0$ of wallets of atomic tokens as a structure encapsulating $\code T \to_0 \mathbb R_{\geq 0}$. This definition induces an element $0 \in \code W_0$ such that $\mathrm{supp}\left(0\right) = \emptyset$: this is the \emph{empty wallet}, which enables us to form the type $\code T \to_0 \code W_0$. 
We define the type $\code W_1$ of wallets of minted tokens 
as a structure that encapsulates a function $\code{bal} \in \code T \to_0 \code{W_0}$. 
The intuition is that $\code{bal}\ \tau_0\ \tau_1$ gives the owned amount of the minted token type created by the AMM pair with tokens $\tau_0$ and $\tau_1$.
Consistently, the function \code{bal} must satisfy two conditions: 
$\code{bal}\ \tau_0\ \tau_1 = \code{bal}\ \tau_1\ \tau_0$, 
meaning that the order of atomic tokens is irrelevant, 
and $\code{bal}\ \tau\ \tau = 0$, meaning that the two token types 
in an AMM must be distinct. 
Our definition of $\code W_1$ encapsulates proofs of these two properties, called \code{unord} and \code{distinct}, respectively.

We map users to their wallets with the types $\code S_0$ and $\code S_1$,
which account for the atomic tokens and for the minted tokens,
respectively.
Finally, we formalize sets of AMM pairs with the type \code{AMMs}. 
The definition is strikingly similar to that of $\code W_1$, but with a changed constraint. 
The intuition is that $\code{res}\ \tau_0\ \tau_1$ gives the reserves of $\tau_0$ in the AMM pair $(\tau_0$, $\tau_1)$,
while  $\code{res}\ \tau_1\ \tau_0$ gives the reserves of $\tau_1$ in the same AMM pair. 
For uninitialized AMM pairs, the obtained reserves must be $0$. 
The property \code{posres} ensures that either an AMM pair has 
no reserves of both token types (\ie, the AMM has not been created yet),
or both token types have strictly positive reserves
(\ie, one cannot deplete the reserves of a single token type in an AMM pair).
We combine the previous definitions in the type $\Gamma$, which represents the state of a blockchain (note that $\Gamma$ abstracts from all the details immaterial for AMMs). 

\paragraph*{Token supply}

Given a blockchain state $s \in \Gamma$, we define the supply of an atomic token type $\tau_0$ as:
\[
\code{atomsupply}_s\left(\tau_0\right) 
\;\; = \;\; 
\sum_{a \in \mathrm{supp}(s\code{.atoms})} 
\hspace{-15pt}
(s\code{.atoms} \; a) \; \tau_0 
\;\; + \;\;
\sum_{\tau_1 \in \mathrm{supp}(s\code{.amms} \, \tau_0)} 
\hspace{-20pt}
(s\code{.amms} \; \tau_0) \; \tau_1
\]
where the partial application $s\code{.amms} \; \tau_0$ 
gives a map with all AMM pairs with $\tau_0$ as one of their token types.
We define the supply of a minted token type $(\tau_0,\tau_1)$ as follows:
\[
\code{mintsupply}_s\left(\tau_0, \tau_1\right) 
\;\; = \;\; 
\sum_{a \in \mathrm{supp}(s\code{.mints})} 
\hspace{-15pt}
(s\code{.mints} \; a) \; \tau_0 \;\tau_1
\]

The corresponding Lean definitions (in \Cref{list:supplydefs}) 
have been split in order to facilitate theorem proving and, in particular, the use of Lean's simplifier.

\begin{table}[t]
    \caption{Fundamental Lean definitions for the state of an AMM system.}
    \label{tab:amm-system}
    \vspace{-20pt}
    \centering
    \begin{tabular}{p{0.25\linewidth}p{0.30\linewidth}p{0.45\linewidth}}
    \begin{lstlisting}[numbers=none]
structure A where
  n: ℕ

structure T where
  n: ℕ

structure W₀ where
  bal: T →₀ ℝ≥0
  
structure S₀ where
  map: A →₀ W₀\end{lstlisting} & 
    \begin{lstlisting}[numbers=none]
structure W₁ where
  bal: T →₀ W₀
  unord: ∀ (τ₀ τ₁: T), 
    bal τ₀ τ₁ = f τ₁ τ₀
  distinct: ∀ (τ: T), 
    bal τ τ = 0
  
structure S₁ where
  map: A →₀ W₁ \end{lstlisting} &
    \begin{lstlisting}[numbers=none]
structure AMMs where
  res: T →₀ W₀
  distinct: ∀ (τ: T), 
    res τ τ = 0
  posres: ∀ (τ₀ τ₁: T), 
    res τ₀ τ₁ ≠ 0  ↔ f τ₁ τ₀ ≠ 0

structure Γ where
  atoms: S₀
  mints: S₁
  amms: AMMs\end{lstlisting}
    \end{tabular}
\end{table}

\begin{table}[t]
\caption{Supply of atomic token types and of minted token types.}
\label{list:supplydefs}
\begin{lstlisting}
noncomputable def S₀.supply (s: S₀) (τ: T): ℝ≥0 := s.map.sum (λ _ w => w τ)

noncomputable def S₁.supply (s: S₁) (τ₀ τ₁: T): ℝ≥0 := 
    s.map.sum (λ _ w => w.get τ₀ τ₁)

noncomputable def AMMs.supply (amms: AMMs) (τ: T): ℝ≥0 := 
    (amms.res τ).sum λ _ x => x

noncomputable def Γ.atomsupply (s: Γ) (τ: T): ℝ≥0 := 
    (s.atoms.supply τ) + (s.amms.supply τ)

noncomputable def Γ.mintsupply (s: Γ) (τ₀ τ₁: T): ℝ≥0 := s.mints.supply τ₀ τ₁
\end{lstlisting}
\end{table}
\paragraph*{AMM reserves}

Given a blockchain state $s$ and two token types $\tau_0$, $\tau_1$,
the terms $s.\code{amms} \; \tau_0 \; \tau_1$
and $s.\code{amms} \; \tau_1 \; \tau_0$
denote, respectively, the reserves of $\tau_0$ and $\tau_1$
in the AMM pair $(\tau_0,\tau_1)$.
This way of accessing the AMM reserves is a bit impractical:
when writing proofs, using $s.\code{amms} \; \tau_0 \; \tau_1$
carries an obligation to provide the functions $\code{distinct}$ and $\code{posres}$.
In particular, this requires the user to explicitly add, in any theorem
using the reserves, the assumption that the reserves
are strictly positive to indicate the AMM pair has been created.
Furthermore, this way of accessing the reserves hides the
fact that when one of the reserves is strictly positive, 
also the other one is such, which again should be made explicit when
writing proofs.

To cope with these issues, 
we build a Lean API that allows for hiding these implementation details (see~\Cref{list:ammapi}). 
For example, given an AMM pair $(\tau_0, \tau_1)$ in a state $s$, 
the expression $s\code{.amms.r_0 \, \tau_0 \; \tau_1 \; init}$ 
give the reserves of token $\tau_0$ in the AMM, 
which are guaranteed to be strictly positive 
under the initialization precondition $\code{init} \in (s\code{.amms.init})$.

\begin{table}[t]
\caption{Fragment of the AMM API: AMM reserves.}
\label{list:ammapi}
\begin{lstlisting}
def AMMs.init (amms: AMMs) (τ₀ τ₁: T): Prop := amms.res τ₀ τ₁ ≠ 0

def AMMs.r₀ (amms: AMMs) (τ₀ τ₁: T) (h: amms.init τ₀ τ₁): ℝ>0 := ⟨ amms.res τ₀ τ₁, 
    by unfold init at h; exact NNReal.neq_zero_imp_gt h ⟩

def AMMs.r₁ (amms: AMMs) (τ₀ τ₁: T) (h: amms.init τ₀ τ₁): ℝ>0 := ⟨ amms.res τ₁ τ₀, 
    by unfold init at h; exact NNReal.neq_zero_imp_gt ((amms.posres τ₀ τ₁).mp h) ⟩
\end{lstlisting}
\end{table}

\paragraph*{Transactions}

Our model encompasses all the main types of transactions supported by AMMs: creating an AMM, adding/removing liquidity, and swapping a token for another. 
Swaps are parameterised by a \emph{swap rate function}, which determines the  exchange rate.
We use the formalization of swap transactions (\Cref{list:swapdefs}) to exemplify the scheme we used for all the transaction types. 

The type $\code{Swap}\left(sx, s, a, \tau_0, \tau_1, x\right)$ represents valid swap transactions in a blockchain state $s$, with the swap rate function $sx$, performed by user $a$ to exchange $x$ amount of the input token $\tau_0$ for a certain amount of the output token $\tau_1$ (\Cref{lst:swapsig}). 
Each element of this type is a structure containing a proof of the validity of the transaction. 
For example, for swap transactions we must prove that the user has enough amount of $\tau_0$ (condition \code{enough}), 
that the AMM pair with tokens $\tau_0$ and $\tau_1$ exists 
(condition \code{exi}), 
and it has enough reserves of $\tau_1$ to give as output 
(condition \code{nodrain}). 
Since the type $\code{Swap}\left(\cdots\right)$
is empty when the parameters do not satisfy the above conditions,
invalid transactions are not really expressible in our model. 
Instead, if $\code{Swap}\left(\cdots\right)$ represents a valid transaction, it will be a singleton type due to proof irrelevance 
(\ie, any two proofs of the same proposition are equal). 

Each transaction is equipped with an $\code{apply}$ function that yields the state reached by executing the transaction in the given state. 
For example, for $sw \in \code{Swap}\left(sx, s, a, \tau_0, \tau_1, x\right)$, 
$\code{apply}\; sw$ yields a state where:
\begin{inlinelist}
\item $a$'s atomic tokens wallet is updated by removing $x$ units of $\tau_0$ and adding $sw\code{.y}$ units of $\tau_1$ (\Cref{lst:updatoms}),
where $sw\code{.y}$ is the amount of tokens outputted by the AMM pair; 
\item accordingly, the AMM reserves are updated by removing $sw\code{.y}$ units of $\tau_1$ and adding $x$ units of $\tau_0$ (\Cref{lst:suby}); 
\item the minted token wallets is unchanged (\Cref{lst:unchangedmints}).
\end{inlinelist}
These definitions use functions and proofs 
not included in \Cref{list:swapdefs} for brevity, 
such as \code{sub} and \code{sub\_r1}. 
These are designed with the same spirit of allowing only valid operations, and so require suitable proofs.
For example, \code{sub\_r1} requires a proof of the existence of the AMM we are removing liquidity from, and a proof that the AMM pair has enough liquidity to retain a positive amount of reserves. 
We build these proofs inline using those contained in the structure:
for instance, the parameter $sw\code{.exi}$ passed to \code{sub\_r1} 
at line 11 is a proof that the AMM pair exists. 
Then, at line 16 we define the constant-product swap rate function, that is the swap rate function used by Uniswap v2. From~\Cref{th:swapdirection} onwards, our results will focus on AMMs using this swap rate function.

\begin{table}[t]
\caption{Definition of the swap transaction type and of its application, as well as the constant-product swap rate function.}
\label{list:swapdefs}
\begin{lstlisting}
abbrev SX := ℝ>0 → ℝ>0 → ℝ>0 → ℝ>0
structure Swap (sx: SX) (s: Γ) (a: A) (τ₀ τ₁: T) (x: ℝ>0) where @\label{lst:swapsig}@
  enough: x ≤ s.atoms.get a τ₀      -- user a has at least x τ₀
  exi: s.amms.init τ₀ τ₁            -- AMM pair τ₀ τ₁ exists in s
  nodrain: x*(sx x (s.amms.r₀ τ₀ τ₁ exi) (s.amms.r₁ τ₀ τ₁ exi))
           < (s.amms.r₁ τ₀ τ₁ exi)  -- AMM has enough output tokens

def Swap.y (sw: Swap sx s a τ₀ τ₁ x): ℝ>0 := 
    x*(sx x (s.amms.r₀ τ₀ τ₁ sw.exi) (s.amms.r₁ τ₀ τ₁ sw.exi))

noncomputable def Swap.apply (sw: Swap sx s a τ₀ τ₁ x): Γ := {
  atoms := (s.atoms.sub a τ₀ x sw.enough).add a τ₁ sw.y, @\label{lst:updatoms}@
  mints := s.mints, @\label{lst:unchangedmints}@
  amms  := (s.amms.sub_r₁ τ₀ τ₁ sw.exi sw.y sw.nodrain).add_r₀ τ₀ τ₁ (by simp[sw.exi]) x } @\label{lst:suby}@

noncomputable def SX.constprod: SX := λ (x r₀ r₁: ℝ+) => r₁/(r₀ + x)
\end{lstlisting}
\end{table}

\paragraph*{Price, networth and Gain}

An important aim of our model is to state and prove economic properties of AMMs related to the networth of their users. The fundamental definitions are in~\Cref{list:networth}. Given a wallet $w \in \code W_0$ and an atomic token price oracle $o \in \code T \to \mathbb R_{>0}$, we define the \emph{value} of $w$ in \cref{lst:atomicworthdef} as:
\[
\code{value}\left(w, o\right) = \sum\limits_{\tau \in \mathrm{supp}\left(w\right)} w\left(\tau\right) \cdot o\left(\tau\right)
\] 
The value of a wallet of minted tokens $w \in \code W_1$
is defined similarly, except that: 
\begin{inlinelist}
\item the summation ranges over $\mathrm{supp}\left(w\code{.u}\right)$, with $w\code{.u} \in \code T^2 \to_0 \mathbb R_{\geq 0}$ representing the uncurrying of $w$;
\item the summation is divided by 2 since, if $\left(\tau_0, \tau_1\right)$ is in the support of $w$, then also $\left(\tau_1, \tau_0\right)$ is in its support.
\end{inlinelist}
For uniformity with the definition of value of atomic wallets, 
also here we assume an oracle that gives the price of (minted) tokens. 
However, while for pricing atomic tokens we indeed resort to an oracle,
for minted tokens this oracle is instantiated to a specific function,
coherently with~\cite{Bartoletti_2022}: 
\[
\code{mintedprice}_s(o, \tau_0, \tau_1) = \frac
    {(s.\code{amms.r_0} \; \tau_0 \; \tau_1) \cdot o(\tau_0) + (s.\code{amms.r_1} \; \tau_0 \; \tau_1) \cdot o(\tau_1)}
    {\code{mintsupply}_s(\tau_0,\tau_1)}
\]
where we have omitted the initialization precondition $\code{h}$ for brevity.

We then define the \emph{networth} of a user as the sum of the value of their two types of wallets (\cref{lst:networthsig}). 
The \emph{gain} of a user upon an update of the blockchain state is the difference between the networth in the new state and that in the old state (\cref{lst:gaindef}).

\begin{table}[t]
\caption{Users' networth and gain.}
\label{list:networth}
\begin{lstlisting}
noncomputable def W₀.value (w: W₀) (o: T → ℝ>0): ℝ≥0 := @\label{lst:atomicworthsig}@
  w.sum (λ τ x => x*(o τ)) @\label{lst:atomicworthdef}@

noncomputable def W₁.value (w: W₁) (o: T → T → ℝ≥0): ℝ≥0 :=
  (w.u.sum (λ p x => x*(o p.fst p.snd))) / 2

noncomputable def Γ.mintedprice (s: Γ) (o: T → ℝ>0) (τ₀ τ₁: T): ℝ≥0 := @\label{lst:mintedpricesig}@
  if h:s.amms.init τ₀ τ₁ then
  ((s.amms.r₀ τ₀ τ₁ h)*(o τ₀)+(s.amms.r₁ τ₀ τ1 h)*(o τ₁)) / (s.mints.supply τ₀ τ₁)
  else 0 -- price is zero if AMM is not initialized

noncomputable def Γ.networth (s: Γ) (a: A) (o: T → ℝ>0): ℝ≥0 := @\label{lst:networthsig}@
  (W₀.value (s.atoms.get a) o) + (W₁.value (s.mints.get a) (s.mintedprice o))

noncomputable def A.gain (a: A) (o: T → ℝ>0) (s s': Γ): ℝ :=
  ((s'.networth a o): ℝ) - ((s.networth a o): ℝ) @\label{lst:gaindef}@
\end{lstlisting}
\end{table}

\paragraph*{Reachable states}

To formalize reachable states, we begin by defining sequences of transactions (\cref{tab:reachable}). 
$\code{Tx}\left(sx,s,s'\right)$ is the type of sequences of valid transactions
(of any kind)
starting from state $s$ and leading to $s'$.
The parameter $sx$ is the swap rate function used in swap transactions.
Technically, $\code{Tx}\left(sx,s,s'\right)$
is an instance of the indexed family of dependent types \code{Tx}, dependent on $sx$ and $s$, and indexed by $s'$. In practice, this means that the constructors must preserve the values of $sx$ and $s$ (building upon the sequence of transactions does not change the swap rate function being used nor the originating state), while $s'$ may change after each construction (the state resulting from the sequence changes with each transaction that is added to it). 
A state $s'$ is \emph{reachable} if there exists a valid sequence of transactions that reaches $s'$ starting from a valid \emph{initial} state $s$, \ie a state with no initialized AMMs or minted token types in circulation.

\begin{table}[t]
    \caption{Sequences of transactions and reachable states.}
    \label{tab:reachable}
    \begin{lstlisting}
inductive Tx (sx: SX) (init: Γ): Γ → Type where @\label{lst:txdef}@
  | empty: Tx sx init init

  | swap (s': Γ) (rs: Tx sx init s')
         (sw: Swap sx s' a τ₀ τ₁ v0):
      Tx sx init sw.apply
  -- Other constructors omitted for brevity
  
def validInit (s: Γ): Prop :=
  (s.amms = AMMs.empty ∧ s.mints = S₁.empty)

def reachable (sx: SX) (s: Γ): Prop :=
  ∃ (init: Γ) (tx: Tx sx init s), validInit init
\end{lstlisting}
\end{table}
\section{Results}
\label{sec:results}
We now present some noteworthy properties of AMMs that we have proven in Lean. 

\Cref{th:ammreachable} ensures that, in any reachable state, there exists an AMM with token types $\tau_0$ and $\tau_1$ if and only if the minted token type $(\tau_0, \tau_1)$ is in circulation, \ie it has a strictly positive supply.
This result showcases the validity of our model with regards to reasoning about reachable states. 
Technically, it also allows us to prove that $\code{mintedprice}_s(o, \tau_0, \tau_1)$ is strictly positive for any initialized AMM pair $(\tau_0, \tau_1)$ in any reachable state $s$.

\begin{proposition}[Existence of AMMs \emph{vs.} minted token supply]
\label{th:ammreachable}
Let $s' \in \Gamma$ be a reachable blockchain state. 
Then, for any minted token type $\left(\tau_0, \tau_1\right)$, its supply in $s'$ is strictly positive if and only if $s'$ has an AMM with token types $\tau_0$ and $\tau_1$.
\end{proposition}
\begin{proof} By induction on the length of the sequence of transactions leading to $s'$:
\begin{itemize}
    
    \item Base case: empty transaction sequence.
    The proof is trivial for both directions, since a valid starting state has no initialized AMMs and no minted tokens in circulation.
    
    \item Inductive case: there are several subcases depending on the last transaction fired in the sequence. 
    Here we consider the creation of the AMM $(\tau_0',\tau_1')$ in reachable state $s'$.
    We proceed by cases on the truth of the equality $\left\{\tau_0, \tau_1\right\} = \left\{\tau_0', \tau_1'\right\}$. If it is a different token pair, then the supply remains unchanged along with the initialization status, and we can conclude by the induction hypothesis. If it is the same token pair, then we just incremented its minted token supply (which is non-negative, so after incrementing it, it must be strictly positive), and we just initialized the AMM.
    
    \item See source code for the other cases. \href{https://github.com/danielepusceddu/lean4-amm/blob/paper/AMMLib/Transaction/Trace.lean#L275}{AMMLib/Transaction/Trace.lean:275} 
    \qedhere
\end{itemize}
\end{proof}

\Cref{th:valueexp} allows us to determine the change in the value of a wallet after it has been updated in some way: the resulting equality is the basis for all the subsequent proofs. To illustrate it, we briefly introduce two definitions that have been omitted before: $\code{drain_0}(w_0, \tau_0)$ is the atomic token wallet such that $\code{drain_0}(w_0, \tau_0)(\tau_0) = 0$ and $\code{drain_0}(w_0, \tau_0)(\tau_1) = w_0(\tau_1)$ for every other token $\tau_1 \neq \tau_0$. We define $\code{drain_1}(w_1, (\tau_0, \tau_0)) \in \code{W_1}$ similarly.
\begin{lemma}[Value expansion]
  \label{th:valueexp}
Let $w_0 \in \code{W}_0$, $o_0 \in \code{T} \to \mathbb R_{>0}$, and $\tau_0, \tau_1 \in \code{T}$. Then, 
$$
\code{value}\left(w_0, o_0\right) \; = \; w_0(\tau_0)\cdot o(\tau_0) + \code{value}\left(\code{drain_0}(w_0, \tau_0), o_0\right)
$$
and, with $w_1 \in \code{W}_1$ and $o_1 \in \code{T^2} \to \mathbb R_{>0}$ such that $o_1(\tau_0, \tau_1) = o_1(\tau_1, \tau_0)$,
$$
\code{value}\left(w_1, o_1\right) \; = \; w_1(\tau_0, \tau_1)\cdot o_1(\tau_0, \tau_1) + \code{value}\left(\code{drain_1}(w_1, \tau_0, \tau_1), o_1\right)
$$  
\end{lemma}
\begin{proof} By definition of value and by properties of the sum over a finite support. Full Lean proof at \href{https://github.com/danielepusceddu/lean4-amm/blob/paper/AMMLib/State/AtomicWall.lean#L116}{AMMLib/State/AtomicWall.lean:116}
\end{proof}

~\Cref{th:swapgain} gives an explicit formula for the gain obtained by a user upon firing a swap transaction. It is fundamental to all proofs involving gain.
\begin{lemma}[Gain of a swap]
\label{th:swapgain}
Let $sw \in \code{Swap}\left(sx,s,a,\tau_0,\tau_1,x\right)$ and 
let $o \in \code{T} \to \mathbb R_{>0}$. Then, 
$$\code{gain}\left(a, o, s, \code{apply} \; sw \right) = \left(sw\code{.y}\cdot o\left(\tau_1\right) - x \cdot o\left(\tau_0\right)\right)
  \cdot \left(1 - \frac{\left(s\code{.mints} \; a \; \tau_0 \; \tau_1\right)}{\code{mintsupply}_s\left(\tau_0,\tau_1\right)}\right)
$$
\end{lemma}
\begin{proof} By repeated application of \Cref{th:valueexp} in order to isolate the value contributed by the token types involved in the swap, and by use of Mathlib's \code{ring\_nf} simplifier tactic. \href{https://github.com/danielepusceddu/lean4-amm/blob/paper/AMMLib/Transaction/Swap/Networth.lean#L55}{AMMLib/Transaction/Swap/Networth.lean:55}
\end{proof}

\Cref{th:swapvsexch} establishes a correspondence between the profitability of a swap transaction (\ie, a positive or negative gain) and the order between the swap rate and the exchange rate given by the price oracle.  
In particular, assuming a trader $a$ who is not a liquidity provider (\ie, $a$ has no minted tokens for the AMM pair targeted by the swap), \Cref{th:swapvsexch} states that:
\begin{itemize}
\item $\code{gain}\left(a, o, s, \code{apply}\;sw\right) < 0 \;\iff\; sx\left(x, r_0, r_1\right) < \nicefrac{o\left(\tau_0\right)}{o\left(\tau_1\right)}$
\item $\code{gain}\left(a, o, s, \code{apply}\;sw\right) = 0 \;\iff\; sx\left(x, r_0, r_1\right) = \nicefrac{o\left(\tau_0\right)}{o\left(\tau_1\right)}$
\item $\code{gain}\left(a, o, s, \code{apply}\;sw\right) > 0 \;\iff\; sx\left(x, r_0, r_1\right) > \nicefrac{o\left(\tau_0\right)}{o\left(\tau_1\right)}$
\end{itemize}
Technically, to formalize this result it is convenient to  use Mathlib's $\code{cmp}$\footnote{\url{https://leanprover-community.github.io/mathlib4_docs/Mathlib/Init/Data/Ordering/Basic.html\#cmp}}, which gives the order between the two parameters.

\begin{lemma}[Swap rate \emph{vs.}~exchange rate]
\label{th:swapvsexch}
Let $sw \in \code{Swap}\left(sx,s,a,\tau_0,\tau_1,x\right)$ be a swap transaction, and let $o \in \code{T} \to \mathbb R_{>0}$ be a price oracle. 
For $i \in \{0,1\}$, let $r_i = s.\code{amms}.r_i\left(\tau_0,\tau_1\right)$. 
If $s.\code{mints}\left(a\right)\left(\tau_0,\tau_1\right) = 0$, then 
$$
\code{cmp}\left(\code{gain}\left(a, o, s, \code{apply}\;sw\right), 0\right) \; = \; \code{cmp}\left(sx\left(x, r_0, r_1\right), \frac{o\left(\tau_0\right)}{o\left(\tau_1\right)}\right)
$$
\end{lemma}
\begin{proof}
By algebraic manipulation and application of ~\Cref{th:swapgain}. \href{https://github.com/danielepusceddu/lean4-amm/blob/paper/AMMLib/Transaction/Swap/Networth.lean#L86}{AMMLib/Transaction/Swap/Networth.lean:86}
\end{proof}

\Cref{th:swapdirection} establishes that, in a constant-product AMM, there exists only one profitable direction for a swap. Namely, if swapping $\tau_0$ for $\tau_1$ gives a positive gain, then swapping in the other direction (\ie, $\tau_1$ for $\tau_0$) will give a negative gain.
Note that the inverse does not hold: a negative gain in a direction does not imply a positive gain in the other direction. 
\begin{lemma}[Unique direction for swap gain]
\label{th:swapdirection}
Let $sw \in \code{Swap}\left(\code{constprod},s,a,\tau_0,\tau_1,x\right)$ and $sw' \in \code{Swap}\left(\code{constprod},s,a,\tau_1,\tau_0,x'\right)$ be two swap transactions in opposite directions, and let $o \in \code{T} \to \mathbb R_{>0}$.
If $\code{gain}\left(a, o, s, \code{apply}\; sw\right) > 0$, then $\code{gain}\left(a, o, s, \code{apply}\; sw\right) < 0$.
\end{lemma}
\begin{proof}
By~\Cref{th:swapvsexch}. \href{https://github.com/danielepusceddu/lean4-amm/blob/paper/AMMLib/Transaction/Swap/Constprod.lean#L160}{AMMLib/Transaction/Swap/Constprod.lean:160}
\end{proof}

\begin{example}
\label{ex:results:swapdirection}
Consider a blockchain state $s$ and atomic tokens $\tau_0$ and $\tau_1$, and assume that: \begin{inlinelist}
\item $a$ is a trader with no minted tokens, \ie $(s\code{.mints}\;a)\;\tau_0\;\tau_1 = 0$;
\item the AMM pair for $(\tau_0, \tau_1)$ has been initialized in $s$ and has reserves $r_0 = (s\code{.amms}\;\tau_0\;\tau_1) = 18$ and $r_1 = (s\code{.amms}\;\tau_1\;\tau_0) = 6$;
\item all the AMMs in $s$ use the constant-product swap rate function;
\item $o$ is a price oracle such that $o\;\tau_0 = 3$ and $o\;\tau_1 = 4$.
\end{inlinelist}
Assume that $a$ wants to sell 6 units of $\tau_1$ for some units of $\tau_0$ with the swap transaction $sw \in \code{Swap}(\code{constprod}, s, a, \tau_1, \tau_0, 6)$. Then, by~\Cref{th:swapgain}, $a$'s gain is given by $9\cdot 3 - 24 = 3 > 0$.
Coherently with~\Cref{th:swapdirection}, any swap in the opposite direction would give a negative gain: \eg, if $a$ sells 6 units of $\tau_0$, her gain would be $3/2 \cdot 4 - 18 = -12$. 
\end{example}

We say that a swap transaction $sw \in \code{Swap}\left(\code{constprod},s,a,\tau_0,\tau_1,x\right)$ is \emph{optimal} for a given price oracle $o$ when, for all $sw' \in \code{Swap}\left(\code{constprod},s,a,\tau_0,\tau_1,x'\right)$ with $x' \neq x$ we have 
\[
\code{gain}\left(a, o, s, \code{apply}\; sw'\right) < \code{gain}\left(a, o, s, \code{apply}\; sw\right)
\]

\Cref{th:suffopt} gives a sufficient condition for the optimality of swaps in constant-product AMMs: it suffices that the ratio of the AMM's reserves after the swap is equal to the exchange rate given by the price oracle.
Intuitively, the condition in~\Cref{th:suffopt} means that the exchange rate between the two token types induced by the AMM (\ie, the ration between the token reserves) is aligned with the exchange rate given by the price oracle, and so further swaps would yield a negative gain.
Note that, by definition, if a swap is optimal then it is also unique, \ie swapping any other amount would yield a suboptimal gain.

\begin{theorem}[Sufficient condition for optimal swaps]
\label{th:suffopt}
Let $sw \in \code{Swap}\left(\code{constprod},s,a,\tau_0,\tau_1,x\right)$
be a swap transaction on a constant-product AMM, 
and let $o \in \code{T} \to \mathbb R_{>0}$ be a price oracle. 
For $i \in \{0,1\}$,
let $r_i' = \left(\code{apply}\; sw\right)\code{.amms.r}_i\left(\tau_0,\tau_1\right)$
be the AMM reserves after the swap.
If $\nicefrac{r_1'}{r_0'} = \nicefrac{o(\tau_0)}{o(\tau_1)}$, then $sw$ is optimal.
\end{theorem}
\begin{proof}
By cases on $x < x'$ and by application of ~\Cref{th:swapvsexch}. {\href{https://github.com/danielepusceddu/lean4-amm/blob/paper/AMMLib/Transaction/Swap/Constprod.lean#L184}{AMMLib/Transaction/Swap/Constprod.lean:184}}
\end{proof}

\begin{example} 
\label{ex:results:suffopt}
Under the assumptions of~\Cref{ex:results:suffopt}, the exchange rate between $\tau_1$ and $\tau_0$ given by the price oracle is
$(o\;\tau_0)/(o\;\tau_1) = 3/4$, while the exchange rate induced by the AMM (\ie, the ratio between the reserves) is $r_1/r_0 = 1/3$.
Hence, to satisfy the equality in~\Cref{th:suffopt} we must perform a swap that increases $r_1$ and decreases $r_0$, \ie $a$ must sell units of $\tau_1$ to buy units of $\tau_0$, coherently with the necessary condition for a positive gain in~\Cref{ex:results:suffopt}. 
\Cref{th:arbitsol} below will establish exactly how many units of $\tau_1$ must be traded.
\end{example}

\Cref{th:arbitsol} gives an explicit formula for the input amount $x$ that yields an optimal swap transaction for a given AMM pair, under the assumption that the user firing the transaction does not hold the AMM's minted token type.
The other implicit assumption is that the user $a$ firing the swap has the needed amount $x$ of units of the sold token, \ie $x \leq (s\code{.atoms}\;a)\;\tau_1$. 
In practice, this assumption can always be satisfied with \emph{flash loans}, which allow $a$ to borrow the amount $x$, perform the swap, and then return the loan in a single, atomic transaction.

\begin{theorem}[Arbitrage for constant-product AMMs]
\label{th:arbitsol}
Let $sw \in \code{Swap}\left(\code{constprod},s,a,\tau_0,\tau_1,x\right)$
be a swap transaction on a constant-product AMM,
and let $o \in \code{T} \to \mathbb R_{>0}$ be a price oracle.
For $i \in \{0,1\}$, let $r_i = s\code{.amms.r}_i\left(\tau_0,\tau_1\right)$ be the AMM reserves. 
If $a$ has no minted tokens (\ie, $s\code{.mints}\left(a\right)\left(\tau_0,\tau_1\right) = 0)$ then $sw$ is optimal if the amount of traded units of $\tau_0$ is:
\[
x \; = \; \sqrt{\frac{o(\tau_1)\cdot r_0 \cdot r_1}{o(\tau_0)}} - r_0
\]
\end{theorem}
\begin{proof}
By algebraic manipulation and~\Cref{th:suffopt}.
\href{https://github.com/danielepusceddu/lean4-amm/blob/paper/AMMLib/Transaction/Swap/Constprod.lean#L316}{AMMLib/Transaction/Swap/Constprod.lean:316}
\end{proof}

\begin{example} 
Under the assumptions of~\Cref{ex:results:suffopt}, we know that to perform a profitable swap the trader must sell units of $\tau_1$, \ie fire a transaction
$\code{Swap}\left(\code{constprod},s,a,\tau_1,\tau_0,x\right)$.
\Cref{th:arbitsol} gives the optimal input value $x$, \ie, the number of sold units of $\tau_1$: 
\[
x \; = \; \sqrt{\frac{3 \cdot 6 \cdot 18}{4}} \, - \, 6 \; = \; 3
\] 
Then, the output amount is given by $sw\code{.y} = 3 \cdot 18/(6+3) = 6$ and, by~\Cref{th:swapgain}, the gain of $a$ is $6 \cdot 3 - 3 \cdot 4 = 6$, which maximizes it. 
Note that, in the new state, the exchange rate given by the AMM coincides with that given by the price oracle: $(r_1 + x)(r_0 - sw\code{.y)} = (o\;\tau_0)/(o\;\tau_1)$.
\end{example}

\section{Related work}
\label{sec:related}

The closest work to ours is~\cite{Nielsen23cpp},
which proposes a methodology for developing and verifying AMMs in the Coq proof assistant. 
In this approach, an AMM is decomposed into multiple interacting smart contract:
\eg, each minted token is modelled as a single smart contract,
following way fungible tokens are encoded 
in blockchains platforms that do not provide custom tokens natively,
as \eg in Ethereum and Tezos.
These smart contracts are then implemented as Coq functions on top of 
ConCert~\cite{Annenkov20cpp}, a generic model of blockchain platforms and smart contracts mechanized in Coq.
Concert is used to verify behavioural properties of smart contracts,
either in isolation or composed with other smart contracts:
this is fundamental for~\cite{Nielsen23cpp}, 
where AMMs are specified as compositions of multiple contracts. 

More specifically, \cite{Nielsen23cpp} applies the proposed methodology to the Dexter2 protocol, which implements a constant-product AMM based on Uniswap v1
on the Tezos blockchain%
\footnote{\url{https://gitlab.com/dexter2tz/dexter2tz/-/tree/master/}}.
The main properties of AMMs proved in~\cite{Nielsen23cpp} are correspondences
between the state of AMMs and the sequences of transactions executed on the blockchain.
In particular, they prove that:
\begin{itemize}
\item the balance recorded in the main AMM contract is coherent with the actual balance resulting from the execution of the sequence of transactions;
\item the supply of the minted token recorded in the main AMM contract is equal to the actual supply resulting from the execution;
\item the state of the minted token contract is coherent with the execution.
\end{itemize}
Furthermore, starting from the Coq specification of the Dexter2 AMM, \cite{Nielsen23cpp} 
extracts verified CamlLigo code, which is directly deployable on the Tezos blockchain. 

Although both our work and \cite{Nielsen23cpp} involve AMM formalizations within a proof assistant, the ultimate goals are quite different. 
The formalization in~\cite{Nielsen23cpp} closely follows the concrete implementation of a particular AMM instance (Dexter2) and produces a deployable implementation 
that is provably coherent with the proposed Coq specification. 
By contrast, we start from a more abstract specification of AMM,
with the goal of studying the properties that must be satisfied by any implementation coherent with the specification.
An advantage of our approach over~\cite{Nielsen23cpp} is that it provides a suitable level of abstraction where proving properties about the economic mechanisms of AMMs,
\ie properties about the gain of users and of equilibria among users' strategies.
In particular, \Cref{th:arbitsol} establishes a paradigmatic property of AMMs, which explains the economic mechanism underlying their design.

Besides these main differences, our AMM model and that in~\cite{Nielsen23cpp}
have several differences.
A notable difference is that our model is based on Uniswap v2, while
\cite{Nielsen23cpp} is based on Uniswap v1.
In particular, this means that our AMMs can handle arbitrary token pairs, while 
in Dexter2 any AMM pairs a token with the blockchain native crypto-currency. 

At the best of our knowledge, \cite{Nielsen23cpp} and ours are currently 
the only mechanized formalizations of AMMs in a proof assistant.
Many other works study economic properties of AMMs that go beyond those proved in this paper. 
The works~\cite{Angeris21analysis} and~\cite{Angeris20aft} study, respectively, an AMM model based on Uniswap similar to ours, and a generalization of the model where the AMM is parameterised over a \emph{trading function} of the AMM reserves, which must remain constant before and after any swap transactions
(in Uniswap v1 and v2, the trading function is just the product between the AMM reserves).
The work~\cite{Bartoletti_2022} studies another generalization of Uniswap v2, where the relation between input and output tokens of swap transactions is determined by an arbitrary \emph{swap rate function}, studying the properties of this function that give rise to a sound economic mechanism of AMMs. 
While both~\cite{Angeris20aft} and~\cite{Bartoletti_2022} share the common goal
of providing general models of AMMs wherein to study their economic behaviour,
they largely diverge on the formalization:
\cite{Angeris20aft} is based on concepts related to convex optimization problems,
while~\cite{Bartoletti_2022} borrows formalization and reasoning techniques from concurrency theory.
The Lean 4 model proposed in this paper follows 
the formalization in~\cite{Bartoletti_2022}, which has the advantage of requiring far less mathematical dependencies: although Mathlib is equipped to reason about convex sets and functions\footnote{\url{https://leanprover-community.github.io/mathlib4_docs/Mathlib/Analysis/Convex/Function}}, it is currently lacking in advanced convex optimization definitions and results as those used in~\cite{Angeris20aft}.

Our AMM model is based on Uniswap v2, one of the most successful AMMs so far.
We briefly discuss some alternative AMM constructions.
Balancer~\cite{balancerpaper} generalizes the constant-product function
used by Uniswap to a constant (weighted geometric) mean
$f(r_1,\cdots,r_n) = \prod_{i=1}^n r_i^{w_i}$, 
where the weight $w_i$ reflects the relevance of a token $\tau_i$
in a tuple $(\tau_1,\cdots,\tau_n)$.
%
Curve~\cite{curvepaper} mixes constant-sum and constant-product functions,
aiming at a swap rate with small fluctuations for large amounts of swapped tokens. 
The work~\cite{Krishnamachari2021fab} studies a variant of the constant-product swap rate invariant, where the rate adjusts dynamically based on oracle prices, 
with the goal of reducing the need for arbitrage transactions.
%
Other approaches aiming at the same goal are studies in~\cite{virtualAMMbalances,Krishnamachari2021fab}, and implemented in~\cite{moonipaper}.
Extending our Lean formalization and results to these alternative AMM designs would require a substantial reworking of our model and proofs.


\section{Conclusions}
\label{sec:conclusions}
In this work we have provided a formalization of AMMs in Lean. Blockchain states are represented as structures containing wallets and AMMs, and transactions as dependent types equipped with a function that defines the state resulting from firing the transaction. Based on this, we have modeled the key economic notions of price, networth and gain. We have then focused on the economic properties of AMMs, constructing machine-checkable proofs. In~\Cref{th:swapgain} we have given an explicit formula for the economic gain of a user after firing a swap transaction. In~\Cref{th:suffopt} we have proved that the rational strategy for traders leads to the alignment between the AMM internal exchange rate and that given by price oracles. Finally, in~\Cref{th:arbitsol} we have derived the amount of tokens that a trader should sell to maximize the gain from a constant-product AMM.

\paragraph*{Design choices}
Before coming up with our Lean formalization, we have experimented with a few alternative definitions.
%
Currently, we use the two pairs of atomic token types $\left(\tau_0,\tau_1\right)$ and $\left(\tau_1, \tau_0\right)$ to represent the same minted token type. 
Initially, we used the type $M$ of sets of atomic tokens of cardinality 2, and modeled wallets of minted tokens by the type $M \to_0 \mathbb R_{\geq 0}$. However, using $M$ in the definition of AMMs turned out not to be as easy. The type $M \to_0 \mathbb R_{\geq 0}^2$ would obviously not work since we would not know which value corresponds to the reserve of which token. On the other hand, the dependent type $\left(m: M\right) \to_0 \code{Option}\left(m \to \mathbb R_{>0}\right)$ would work after defining $0 := \code{None}$, at the cost of losing the straightforward definition for supply of atomic tokens.
We have opted for the custom subtype $\mathbb R_{>0}$ to represent the positive reals, since they simplify writing certain definitions and proof passages (\eg, avoiding the use of garbage outputs in the definitions where negative inputs would not make sense). 
This choice however turned out to have some cons, since it makes using Mathlib more complex, and in some cases we have to coerce back to the reals anyway (\eg when reasoning about the gain).
Using Mathlib's reals would perhaps lead to a smoother treatment.

\paragraph*{Limitations}
Compared to real-world AMM implementations, 
our Lean formalization introduces a few simplifications, that overall contribute to keeping our proofs manageable. Bridging the gap with real AMMs would require several extensions, which we discuss below as directions for future work.
AMMs typically implement a trading fee $\phi \in [0,1]$ that represents the portion of the swap amount kept by the AMM. 
While modeling the fee would be easy
($\phi$ would be an additional parameter to $\code{SX}$ functions, to $\code{Swap}$ types, and to transactions $\code{Tx}$), it would require a major reworking of all the results that deal with swaps. 
Our swaps have \emph{zero-slippage}, in that either a swap gives exactly the amount of  tokens required by a user, or they are aborted. While on the one hand this is desirable (\eg, it rules out sandwich attacks), on the other hand it has drawbacks related to liveness, since a user may need to repeatedly send swap transactions until one is accepted. Real-world AMM implementations allow users to specify a slippage tolerance in the form of the minimum amount of tokens they expect from a swap. Extending our model to encompass slippage tolerance would require to add a parameter to each transaction type, and a minor reworking of the results.
%
Some AMM implementations allow users to create AMMs pairs involving \emph{minted} token types. Consequently, the tokens minted by these AMMs in general are ``nestings'' of token types. Extending our model in this direction would require 
to replace price oracles in our results with a suitable price function.



\bibliography{main}
\end{document}